\pdfoutput=1
\documentclass[a4paper,USenglish,cleveref,numberwithinsect]{lipics-v2021}
\nolinenumbers

\hideLIPIcs
\usepackage{colonequals}
\usepackage{microtype}
\usepackage{bm}
\usepackage{tikz}
\usetikzlibrary{arrows,decorations.pathreplacing,calligraphy,shapes.misc,intersections,calc}
\newcommand{\tikzgrid}[2]{
  \newcommand{\xmin}{0}
  \newcommand{\xmax}{#1}
  \newcommand{\ymin}{0}
  \newcommand{\ymax}{#2}
    \foreach \i in {\xmin,\the\numexpr\xmin + 2\relax, ..., \xmax} {
      \draw[gray, very thin] (\i,\the\numexpr\ymin-1\relax) -- (\i,\the\numexpr\ymax+1\relax);
    }
    \foreach \j in {\ymin,\the\numexpr\ymin + 2\relax, ..., \ymax} {
      \draw[gray, very thin] (\the\numexpr\xmin-1\relax,\j) --
      (\the\numexpr\xmax+1\relax,\j);
    }
}
\newcommand{\wall}[5]{
    \foreach \i in {1,...,#2}
      \foreach \j in {1,...,#1} {
        \node[draw,circle] (a\i\j) at (\i, \j) {};
      }
    \foreach \i in {1,...,#4}
      \foreach \j in {1,...,#1} {
        \draw (a\i\j) -- (a\the\numexpr \i + 1\relax\j);
      }
    \foreach \i in {1,3,...,#2}
      \foreach \j in {1,3,...,#3} {
        \draw (a\i\j) -- (a\i\the\numexpr \j + 1\relax);
      }
    \foreach \i in {2,4,...,#2}
      \foreach \j in {2,4,#5} {
        \draw (a\i\j) -- (a\i\the\numexpr \j + 1\relax);
      }
    }

\tikzset{cross/.style={cross out, draw=black, minimum size=2*(#1-\pgflinewidth),
inner sep=0pt, outer sep=0pt, rotate=45, thick},
cross/.default={1pt}}

\hypersetup{
    colorlinks,
    linkcolor={red!50!black},
    citecolor={blue!50!black},
    urlcolor={blue!30!black}
}

\title{Degree-3 planar graphs as topological minors of wall graphs in polynomial time}
\author{Antoine Amarilli}{LTCI, Télécom Paris, Institut Polytechnique de Paris,
France \and
\url{https://a3nm.net/}}{a3nm@a3nm.net}{https://orcid.org/0000-0002-7977-4441}{}
\authorrunning{Antoine Amarilli}
\Copyright{Antoine Amarilli}

\ccsdesc[500]{Mathematics of computing~Discrete mathematics~Graph theory}
\keywords{Planar graph, topological minor, wall graph}

\begin{document}

\maketitle

\begin{abstract}
  In this note, we give a proof of the fact that we can efficiently find degree-3 planar graphs as
  topological minors of sufficiently large wall graphs. The result is needed as
  an intermediate step to fix a proof in my PhD
  thesis~\cite{amarilli2016leveraging}.
\end{abstract}

\section{Introduction and related work}
\label{sec:intro}

In this note, we explain how, given a degree-3 planar graph, we can 
find it as a topological minor of a wall graph, and do so efficiently, i.e., in
polynomial time. This result is needed as
  an intermediate step to fix a proof in my PhD
  thesis~\cite{amarilli2016leveraging}. Here is the formal statement of the
  result:

  \begin{theorem}
    \label{thm:main}
  Given as input any
    degree-3 planar graph~$G$ with~$n$ vertices, we can compute in time
    $O(n^{10})$ a wall
    graph~$H$ of size $O(n^4)$ by $O(n^4)$ and
    an embedding of~$G$ as a topological minor of~$H$.
  \end{theorem}

  Note that, in this result, we make no effort to optimize the degree of the
  polynomial.

  The proof presented here uses the well-known
  fact~\cite{schnyder1990embedding,chrobak1995linear}
  that planar graphs can be drawn in linear time on a grid with integer
  coordinates. Specifically, it works by starting from the drawing and adjusting
  it to obtain the embedding.

  \subparagraph*{Related work.}
  An alternative route to show Theorem~\ref{thm:main} was pointed out to us
  in~\cite{cstheory} after we posted a first version of this note. Specifically,
  it is shown in~\cite{tamassia1989planar} that we can embed an input planar graph
  of maximal degree 4 into a grid in linear time, and the \emph{grid embedding}
  computed in~\cite{tamassia1989planar} is a topological embedding (satisfying
  additional properties). The embedding can then be converted to a topological
  embedding into a wall graph, thanks to the fact that the input graph is
  degree-3.  This proof method is different from the one which we present in
  this note; further, it would achieve a better bound, and is arguably simpler.

  Similar results to Theorem~\ref{thm:main} are already known to embed planar
  graphs as minors of grid graphs \cite[Result~1.5]{robertson1994quickly}, but
  we are not aware of a discussion of the complexity of this process. In
  particular, the proof in~\cite{robertson1994quickly} does not seem easy to
  adapt to a polynomial-time process, as it involves, e.g., the computation of
  Hamiltonian circuits.

  Theorem~\ref{thm:main} is not related to the celebrated theorem by Chekuri and
  Chuzhoy~\cite{chekuri2016polynomial}, which shows that we can embed a grid in
  any graph with sufficiently large but polynomial treewidth. In our result, the
  graph that we embed is an arbitrary planar degree-3 graph, but importantly the
  target of the embedding is always a wall graph. This is why we are also able
  to show a deterministic PTIME bound, whereas the algorithm
  of~\cite{chekuri2016polynomial} runs in randomized PTIME.

\subparagraph*{Acknowledgements.} I am grateful to Pierre Bourhis, Mikaël Monet,
and Pierre Senellart for discussions about the problem, in particular I warmly thank Mikaël
Monet for proofreading the note in detail, identifying some problems, and
discussing possible solutions.

\section{Preliminaries and result statement}
\label{sec:prelim}

\begin{figure}
  \centering
  \begin{tikzpicture}[xscale=1.5]
    \wall{6}{8}{5}{7}{4}
  \end{tikzpicture}
  \caption{The $(6,8)$-wall}
  \label{fig:wall}
\end{figure}

An (undirected) \emph{graph} $G = (V, E)$ consists of a set of \emph{vertices}~$V$
and a set of \emph{edges}~$E$ which are pairs of distinct vertices that are said
to be \emph{adjacent}. The graph is
\emph{degree-3} if the \emph{degree} of each vertex, i.e., the number of edges
in which it appears, is at most~$3$. The graph is \emph{planar} if it can be
drawn on the plane without edge crossings. A \emph{path} connecting~$u \in V$
and $v \in V$ in~$G$ is a sequence of vertices $u = w_1, \ldots, w_n = v$ such
that $w_i$ and $w_{i+1}$ are adjacent for all $1 \leq i \leq n$. The \emph{inner
vertices} of the path are $w_2, \ldots, w_{n-1}$: there may be none if the
\emph{length}~$n$ of the path is~$0$ (i.e., $u = v$) or~$1$ (i.e., $u$ and~$v$
are adjacent and the path directly takes that edge).

We define \emph{wall graphs} following~\cite{dragan2011spanners}.
For integers $r$ and $s$, the \emph{$(r,s)$-wall} is
the graph $W_{r,s}$ with vertices $\{(i, j) \mid 1 \leq i \leq r, 1 \leq j \leq
s\}$ and edges $\{\{(i,j), (i,j+1)\} \mid 1 \leq i \leq r, 1 \leq j < s\}$ and
$\{\{(i,j), (i+(-1)^{i+j},j) \mid 1 \leq i < r, 1 \leq j \leq s\}$. An example wall graph is given on
Figure~\ref{fig:wall}. Note that wall graphs are always degree-3 and planar.

A \emph{(topological) embedding} of a graph~$G = (V, E)$ into a graph~$H = (V',
E')$ consists of an injective function $f\colon V\to V'$ and a function $g$
mapping each edge $\{u, v\}$ of~$E$ to a path $g(\{u, v\})$ in~$H$ connecting $f(u)$
and~$f(v)$ such that all these paths are vertex-disjoint, i.e., for any two
edges $\{u, v\}$ and $\{u', v'\}$ of~$E$, the set of inner vertices of the paths
$g(\{u, v\})$ and $g(\{u', v'\})$ are disjoint.

Our goal is to show Theorem~\ref{thm:main}.
The proof proceeds in three steps.  Some of the details are a bit tedious, and the
goal of this note is to make them precise, but the
overall strategy is rather simple.

First, we use the result
of~\cite{schnyder1990embedding} to obtain in linear time a drawing of the input graph with integer
coordinates. The idea will be to convert this to an embedding in a sufficiently
large wall graph, by scaling up the drawing. We also show some geometric
lemmas on the drawing to bound how close the elements of the drawing can be: how
close can a vertex be to a segment to which it does not belong, and how close two
segments sharing a common endpoint can be when they are sufficiently far away
from the common endpoint.

Second, we explain how edges can be translated to paths in the wall graph, by
showing that any straight line between points with integer coordinates can be
``approximated'' by a path in a wall graph which ``stays close'' to the line.

Third, we use the approximate paths to transform the drawing of the graph into
something that resembles a topological embedding except that the paths may
overlap close to their endpoints.

Fourth, we explain how to modify the embedding
close to the vertices to fix this problem.

\section{Drawing the input graph~$G$ and showing lemmas about the drawing}
\label{sec:drawing}

Given the input graph~$G = (V, E)$, let~$n$ be its number of vertices. The problem is
trivial if $n \leq 2$, so we assume $n \geq 3$. We know by the work of
Schnyder~\cite{schnyder1990embedding} that we can compute in linear time a
\emph{straight line embedding} of~$G$ in the $n-2$ by $n-2$ grid. This implies
that we have an injective \emph{drawing function}
$\delta\colon V \to \{0, \ldots, n-1\}^2$
mapping each vertex $u \in V$ to a point with integer coordinates in this range, so that
the edges do not cross, i.e., for any two edges $e = \{u,v\}$ and $e' =
\{u',v'\} \in E$ with $e \neq e'$,
the segments $[\delta(u), \delta(v)]$ and $[\delta(u'), \delta(v')]$ do not
cross (except that they may share the same endpoints if some vertices among $u$,
$v$, $u'$, and $v'$ are equal).

We will want to replace the segments of the drawing by paths in a suitable wall
graph. To do this and ensure that the paths do not intersect, it will be
important to understand some properties of the drawing.
First, we must understand how close a segment between points with integer
coordinates can pass by another point with integer coordinates which is not on
the segment. For this, we show a lemma:

\begin{figure}
  \begin{subfigure}{.5\linewidth}
  \begin{tikzpicture}[scale=.4]
    \tikzgrid{16}{18}
    \node[cross=2pt] (p) at (2,2) {};
    \node (pl) at (1.5,1.5)  {$p$};
    \node[cross=2pt] (r) at (16,14) {};
    \node (rl) at (15.5,14.5)  {$r$};
    \node[cross=2pt] (q) at (6,6) {};
    \node (ql) at (5.5,6.5)  {$q$};
    \node (qq) at (6.5,5.4) {};
    \node (qqq) at (5.9,6.1) {};
    \draw[very thick,orange] ($(6.3,5.65)+(-.02,.02)$)--($(q)+(-.02,.02)$);
    \node[cross=2pt] (q2) at (6,6) {};
    \draw (p) -- (r);
  \end{tikzpicture}
  \caption{Situation of Lemma~\ref{lem:pass}. The point of the lemma is to give
    a lower bound on the distance
    between $q$ and the segment $[p,r]$ (in orange)}
  \label{fig:pass}
  \end{subfigure}
  \hfill
  \begin{subfigure}{.45\linewidth}
  \begin{tikzpicture}[scale=1]
    \tikzgrid{4}{4}
    \draw[very thick, orange] (1,2) -- (3,2);
    \draw[very thick, orange] (2,1) -- (2,3);
    \node[cross=4pt,very thick] (q) at (2,2) {};
    \draw[blue,dashed] (.5,-.5) -- (4.5,3.5);
    \draw[blue,very thick] (2,2) -- (2.5,1.5);
  \end{tikzpicture}
  \caption{Argument used in the proof of Lemma~\ref{lem:pass}. The orange area
    corresponds to the point by which a segment cannot pass, extending to length
    $1/n$ on each side of the point. The dashed blue
    line is the closest that the segment may pass, and the distance bound in thick blue
    is the one shown in the lemma, namely, $1/(\sqrt{2} n)$}
  \label{fig:pythagoras}
  \end{subfigure}
  \caption{Illustrations for Lemma~\ref{lem:pass}}
  \label{fig:lempass}
\end{figure}
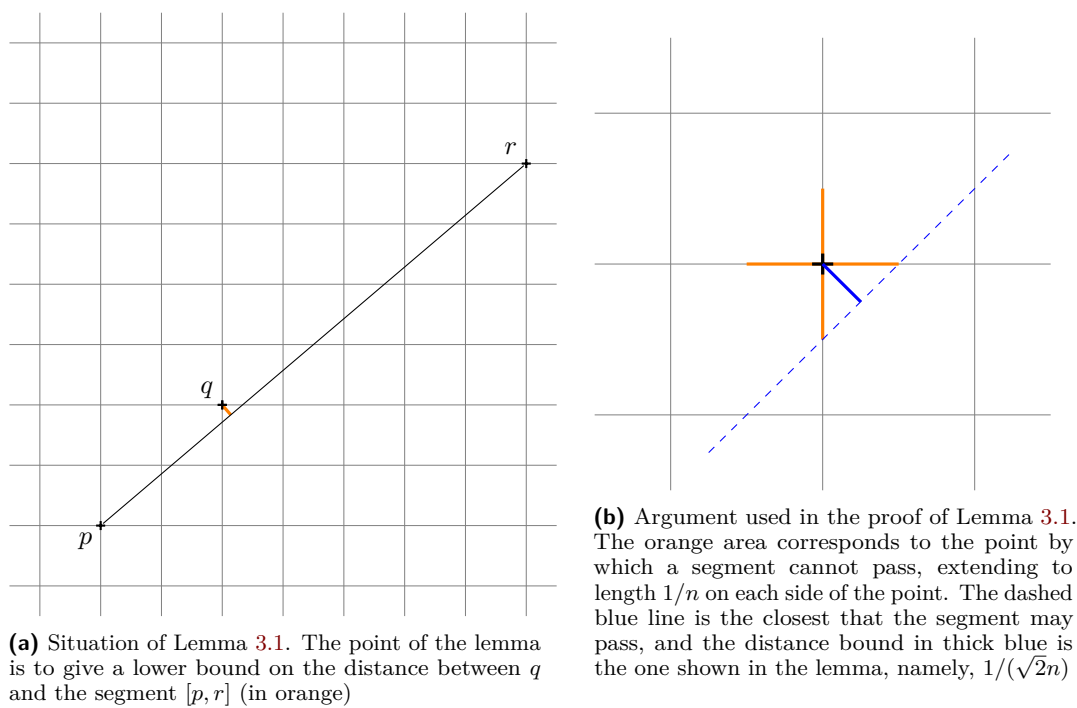

\begin{lemma}
  \label{lem:pass}
  Let $p, q, r$ be pairwise distinct points with integer coordinates between~$0$ and~$n-1$. Assume
  that the point $q$ is not on the segment~$[p,r]$. Then the distance
  between~$q$ and the segment $[p,r]$ is at least $\frac{1}{\sqrt{2} n}$.
\end{lemma}

The lemma is illustrated as Figure~\ref{fig:pass}.

\begin{proof}
  If the segment is parallel to one of the axes, then the result is trivial
  because its distance to points with integer coordinates not in the segment is
  clearly at least~$1$, hence at least~$\frac{1}{\sqrt{2} n}$. Hence, we assume that the
  segment is not parallel to any axis, hence the first coordinates of~$p$
  and~$r$ are different, and the second coordinates of~$p$ and~$r$ are also different.

  We first show the following claim (*): for all points of the segment $[p,r]$ having exactly one
  integer coordinate in the range $\{0, \ldots, n\}$, the difference between the other coordinate and the
  nearest integer is at least $1/n$.

  Indeed, let $(a, b)$ and $(c, d)$ be the respective coordinates of~$p$
  and~$r$. 
  The points in the segment are those with coordinates $(a + x(c-a), b + x(d-b))$
  for $0 \leq x \leq 1$. Assume that some coordinate is integer, say the first,
  then $a + x(c-a)$ is an integer, say $e$. Hence, $x = \frac{e-a}{c-a}$,
  recalling that $c\neq a$, where $a$ and $e$ and $c$ are integers in $\{0,
  \ldots, n-1\}$. Then the second coordinate is $b + x(d-b)$. This can be
  expressed as a fraction with integer numerator and with denominator $c-a$, which
  is an integer of absolute value at most~$n$. Hence, if the value of the second
  coordinate is not an
  integer, its distance to the nearest integer is at least $1/n$.
  The same argument applies if we exchange the role of the first and second
  coordinates. This establishes claim (*) above.

  Now, we use convexity. Considering the point $q$, the segment does not go via
  that point, so considering all points with one integer coordinate matching
  that of~$q$ and the other integer coordinate differing from the coordinate
  of~$q$ by strictly less than $1/n$, we know that the segment cannot go via
  these points. So the closest the segment can go is $\frac{1}{\sqrt{2}} \times
  \frac{1}{n}$ (see Figure~\ref{fig:pythagoras} for an illustration), concluding
  the proof.
\end{proof}

Second, we need to understand how close two segments with integer coordinates
that share a common endpoint can be from one another. Of course, close to the
endpoint, they can be arbitrarily close; but we need to show that they are
sufficiently far away once we are sufficiently far from the common endpoint.
For this, let us define the notion of \emph{box}:

\begin{definition}
  \label{def:box}
  Given a point $p$ and length $\epsilon > 0$, the \emph{box} of radius
  $\epsilon$ centered on~$p$ is the square region with center~$p$ and sides $2
  \epsilon$.
\end{definition}

We can then claim:

\begin{lemma}
  \label{lem:boxdist}
  Let $p, q, r$ be pairwise distinct points with integer coordinates, and consider the segments
  $[p,q]$ and $[p,r]$. We assume that $[p,q]$ and $[p,r]$ are not collinear,
  so $p$ is their only intersection.

  Let $\epsilon > 0$, and consider the box of radius~$\epsilon$
  centered on~$p$. Let $q'$ and $r'$ be arbitrary points of $[p,q]$ and
  $[p,r]$ respectively such that both are outside the box. Then
  the distance between $q'$ and $r'$ is
  at least $\frac{\epsilon}{2 n^2}$.
\end{lemma}

We illustrate the lemma in Figure~\ref{fig:boxdist}.
We need a simple auxiliary claim to show the result.

\begin{figure}
  \centering
  \begin{tikzpicture}[scale=1, extended line/.style={shorten >=-0cm,shorten
    <=-1cm}]
    \tikzgrid{10}{10}
    \node[cross=2pt] (p) at (2,2) {};
    \node (pl) at (1.7,1.7)  {$p$};
    \node[cross=2pt] (q) at (4,8) {};
    \node (ql) at (4.5,8.5)  {$q$};
    \node[cross=2pt] (r) at (10,8) {};
    \node (rl) at (10.5,8.5)  {$r$};
    \draw[name path=pr] (p) -- (r);
    \draw[name path=pq] (p) -- (q);
    \draw [blue,name path=rect] (1.3,1.3) rectangle (2.7,2.7);
    \draw [name path=circ] (2,2) circle (0.7);
    \draw[to-to] (1.3,1.2) -- (2.7,1.2);
    \node (eps) at (2.3,.9) {$2\epsilon$};
    \path [name intersections={of=circ and pq,by=qq}];
    \path [name intersections={of=circ and pr,by=rr}];
    \draw [orange,very thick,to-to] (qq) -- (rr);
    \node[cross=2pt] (qqq) at (qq) {};
    \node[cross=2pt] (rrr) at (rr) {};
    \node (qql) at (2.2,3) {$q''$};
    \node (qql) at (3,2.2) {$r''$};
    \draw [dashed,name path=para,extended line=1cm] (q) --
    +($(rrr)+(rrr)+(rrr)+(rrr)+(rrr)+(rrr)+(rrr)+(rrr)+(rrr)+(rrr)+(rrr)+(rrr)-(qqq)-(qqq)-(qqq)-(qqq)-(qqq)-(qqq)-(qqq)-(qqq)-(qqq)-(qqq)-(qqq)-(qqq)$);
    \path [name intersections={of=para and pr,by=rrrr}];
    \node[cross=2pt] (rrrrr) at (rrrr) {};
    \node (rrrrl) at (7,6.5) {$r'''$};
  \end{tikzpicture}
  \caption{Illustration of Lemma~\ref{lem:boxdist}. The lemma is shown by
  establishing 
  a lower bound on the orange distance, depending on the size of the box
  (in blue).}
  \label{fig:boxdist}
\end{figure}
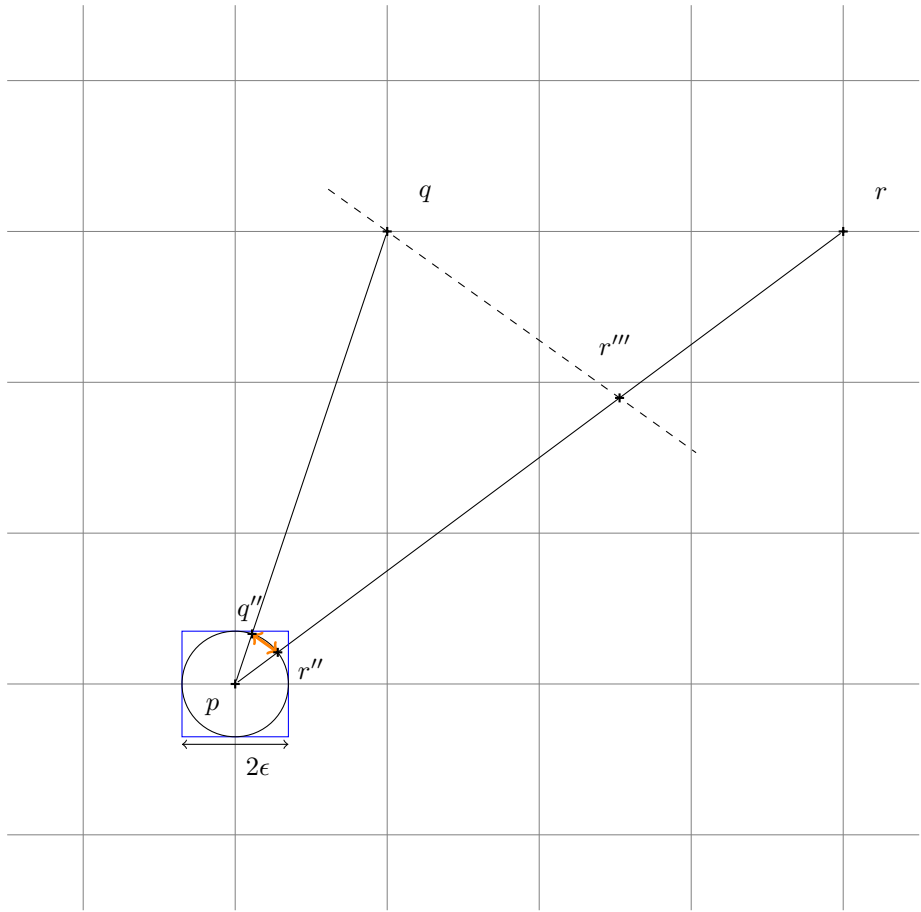

\begin{claim}
  \label{clm:points}
  Consider a triangle $p$, $q''$, $r''$, where the angles on $q''$ and on $r''$ are
  acute. Let $q$ be a point of the ray with direction
  $(p,q'')$
  starting at~$q''$ 
  and let $r$ be a point on the ray with direction
  $(p,r'')$
starting at~$r''$.
  Then the distance between $q$ and $r$ is greater than or equal to
  the distance between $q''$ and $r''$.
\end{claim}

The claim is illustrated on Figure~\ref{fig:points}. Let us prove it:

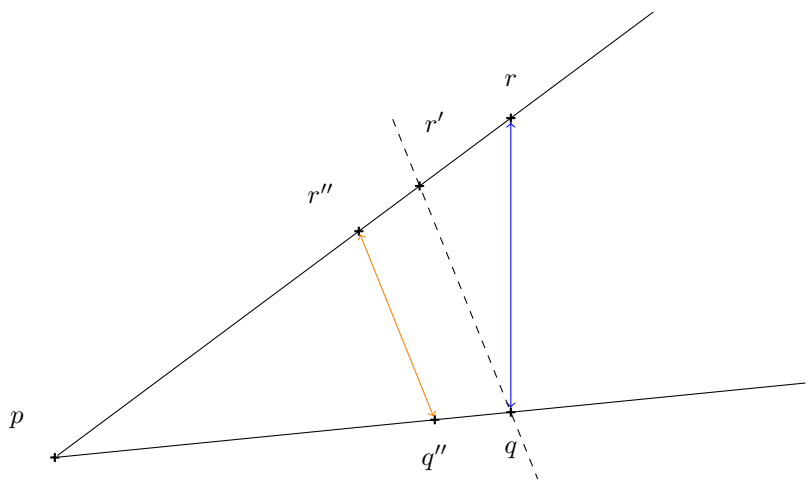
\begin{figure}
  \begin{tikzpicture}[extended line/.style={shorten >=-1cm,shorten <=-1cm}]
    \node[cross=2pt] (p) at (0, 0) {};
    \node (pl) at (-.5, .5) {$p$};
    \node[cross=2pt] (q) at (5, .5) {};
    \node[cross=2pt] (r) at (4, 3) {};
    \node[cross=2pt] (q2) at (6, .6) {};
    \node[cross=2pt] (r2) at (6, 4.5) {};
    \node[cross=2pt] (r3) at (4.8, 3.6) {};
    \node (qq) at (10, 1) {};
    \node (rr) at (8, 6) {};
    \draw (p) -- (qq);
    \draw[orange,to-to] (r) -- (q);
    \draw[blue,to-to] (r2) -- (q2);
    \draw (p) -- (rr);
    \node (ql) at (5, 0) {$q''$};
    \node (rl) at (3.5, 3.5) {$r''$};
    \node (ql2) at (6, 0.1) {$q$};
    \node (rl2) at (6, 5) {$r$};
    \node (rl3) at (5, 4.45) {$r'$};
    \draw[dashed,extended line = 1cm] (r3) -- (q2);
  \end{tikzpicture}
  \caption{Illustration of Claim~\ref{clm:points}: the distance in blue is
  greater than or equal to the distance in orange, provided the angles on $r''$
  and $q''$ are acute. We also illustrate the first step of the proof.}
  \label{fig:points}
\end{figure}

\begin{proof}
  We first reduce to the case where $q = q''$.
  This step is illustrated on Figure~\ref{fig:points}.
  Consider the line $(q'', r'')$.
  Up to swapping the points $q$ and $q''$ with $r$ and $r''$, without loss of
  generality we can assume that the distance of point $q$ to that line is no
  greater than the distance of point $r$ to that line.
  
  Let $r'$ the
  intersection with $(p,r)$ of the parallel line to $(q'', r'')$. It is
  immediate (e.g., by Thales's theorem) that the distance between $q$ and $r'$
  is no smaller than the distance between $q''$ and $r''$. Further, considering
  the triangle $p$, $q$, $r'$, the angles on $q$ and $r'$ in that triangle are
  respectively equal to the angles on $q''$ and $r''$ on the original triangle,
  so they are also acute; further the point $r$ is still on the ray starting at $r'$
  with direction $(p,r')$. Hence, it suffices to study the case where $q = q''$.
  (This case is not illustrated.)

  Now, in this case, since the angle on $r''$ is acute, we know that the
  projection of $q''$ on $(p,r'')$ is in the segment $[p,r'']$. This means that,
  letting $r'$ be this projection, and moving along the ray starting at $r'$
  with direction $(p,r')$, the distance to $q''$ is not decreasing. Doing so, we
  first encounter $r''$ and then $r$, so we conclude that indeed the distance
  between $q''$ and $r''$ is no greater than the distance between $q=q''$ and $r$.
\end{proof}

We now prove Lemma~\ref{lem:boxdist}:

\begin{proof}[Proof of Lemma~\ref{lem:boxdist}]
  Let $q''$ and $r''$ be the points on $[p,q]$ and $[p,r]$ respectively at
  distance $\epsilon$ from $p$. These two points are inside the box. As the
  triangle $p$, $q''$, $r''$ is isosceles, the angles on $q''$ and $r''$ must be
  acute. Thus, to show the bounds on points $p'$ and $q'$ as in the lemma
  statement, by Claim~\ref{clm:points} it suffices to show the claim on
  the distance between $q''$ and $r''$.

  First, up to exchanging $q$ and $r$ if necessary, we assume that the line
  parallel to $(q'',r'')$ going through $q$ intersects the line $(p,r)$ within the
  segment $[p,r]$. Let $r'''$ be the intersection point. Considering the point
  $q$ and the segment $[p,r]$, we have by Lemma~\ref{lem:pass} a lower bound of
  $\frac{1}{\sqrt{2} n}$ on the distance between $q$ and $[p,r]$, translating to
  the same lower bound on the distance between $q$ and the point $r'''$ which
  belongs to $[p,r]$.

  We now use Thales's theorem to show that the ratio between the lengths of
  $[p,q'']$ and $[p,q]$ is the same as the ratio between the lengths
  of~$[q'',r'']$ and $[q,r''']$. From the lower bound of
  Lemma~\ref{lem:pass}, letting 
  $x$ be the distance between $p$ and $q$, we have:
  \[
    \frac{\epsilon}{x} \leq d \sqrt{2} n
  \]
  We deduce:
  \[
    d \geq \frac{\epsilon}{x \sqrt{2}n}
  \]
  Now, $x \leq \sqrt{2} n$ given the range of coordinates, so we deduce:
  \[
    d \geq \frac{\epsilon}{2 n^2}
  \]
  Thus, we have shown our lower bound on the distance between $q''$ and $r''$,
  which as we had argued implies a lower bound on point $q'$ and $r'$ as in the
  lemma statement.
\end{proof}

\section{Approximating lines with paths}

Our idea is to start from the drawing function defined in the previous section
and transform it to a
topological embedding into a sufficiently large wall graph. To do this, we will first argue how each
straight line segment in the drawing can be transformed to a path in a wall graph which stays
sufficiently close to the line. Formally, we show:

\begin{figure}
  \begin{tikzpicture}
    \wall{8}{14}{7}{13}{6}
    \draw (3,8) -- (13,2);
    \node (n1l) at (3,8.5) {$p_0$};
    \node (n1l) at (5.3,7.35) {$p_1'$};
    \node (n1l) at (6.3,6.35) {$p_2'$};
    \node (n1l) at (6.7,5.35) {$p_3'$};
    \node (n1l) at (8,5.35) {$p_3$};
    \node (n1l) at (10.3,4.35) {$p_4'$};
    \node (n1l) at (11.3,3.35) {$p_5'$};
    \node (n1l) at (11.7,2.35) {$p_6'$};
    \node (n1l) at (13,2.35) {$p_6$};
    \node[draw,circle,blue] (n1) at (3,8) {};
    \node[draw,circle,blue] (n2) at (4,8) {};
    \node[draw,circle,blue] (n3) at (5,8) {};
    \node[draw,circle,blue] (n4) at (5,7) {};
    \node[draw,circle,blue] (n5) at (6,7) {};
    \node[draw,circle,blue] (n6) at (6,6) {};
    \node[draw,circle,blue] (n7) at (7,6) {};
    \node[draw,circle,blue] (n8) at (7,5) {};
    \node[draw,circle,blue] (n9) at (8,5) {};
    \node[draw,circle,blue] (n10) at (9,5) {};
    \node[draw,circle,blue] (n11) at (10,5) {};
    \node[draw,circle,blue] (n12) at (10,4) {};
    \node[draw,circle,blue] (n13) at (11,4) {};
    \node[draw,circle,blue] (n14) at (11,3) {};
    \node[draw,circle,blue] (n15) at (12,3) {};
    \node[draw,circle,blue] (n16) at (12,2) {};
    \node[draw,circle,blue] (n17) at (13,2) {};
    \draw[thick,blue] (n1) -- (n2) -- (n3) -- (n4) -- (n5) -- (n6) -- (n7) -- (n8) --
    (n9) -- (n10) -- (n11) -- (n12) -- (n13) -- (n14) -- (n15) -- (n16) --
    (n17);
  \end{tikzpicture}
  \caption{Approximating a segment by a path in a wall graph. The points $p_i$
  are equal to $p_i'$ unless otherwise written.}
  \label{fig:wallapprox}
\end{figure}
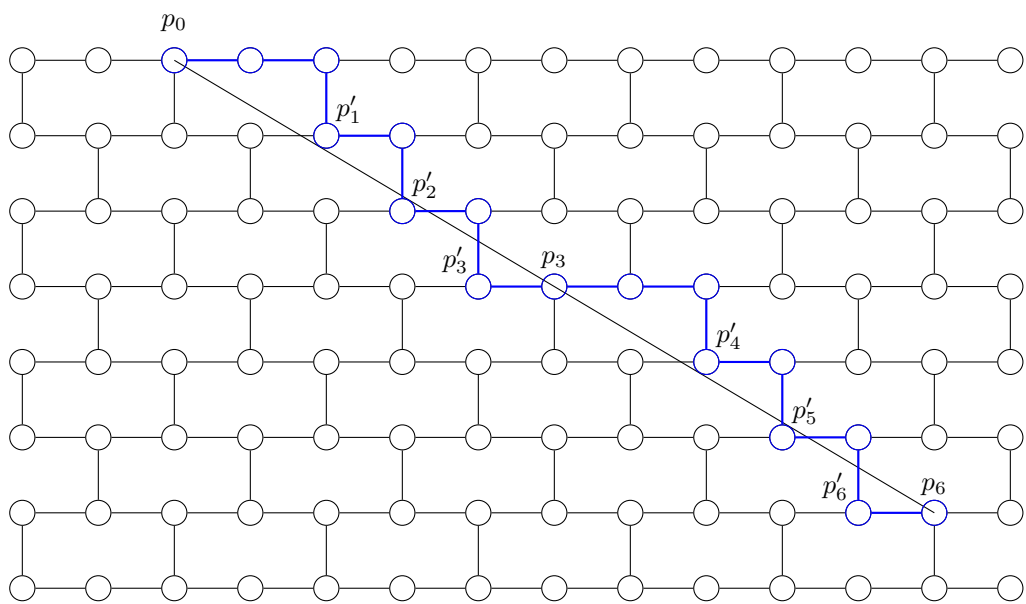

\begin{lemma}
  \label{lem:line}
  Let $r, s$ be two integers, and consider two distinct integer points $(i,j)$
  and $(i',j')$ that are vertices of the $(r,s)$-wall $W_{r,s}$, with $1 \leq i, i' \leq r$ and
  $1 \leq j, j' \leq s$. Consider the segment $[(i,j), (i',j')]$. Then we can
  construct in~$O(r \times s)$ a simple path from~$(i,j)$ to~$(i',j')$ in~$W_{r,s}$
  such that, for each traversed vertex, its distance to the segment is at
  most~$2$.
\end{lemma}

  The lemma is illustrated on Figure~\ref{fig:wallapprox}.

\begin{proof}
  The claim is immediate if the segment is horizontal, i.e., $i = i'$, so we
  focus on the case of a non-horizontal segment. Up to swapping both points, we
  assume that $i' > i$.
  We consider points $p_0, \ldots, p_n$ with $n = i'-i$, i.e., $n > 0$, defined
  in the following way: $p_0$ is $(i,j)$, and for $k > 0$ the point $p_k$ is
  some point of the wall graph with first coordinate $i+k$ which is at minimal
  distance of the point $\pi_k$ of the segment having this first coordinate, i.e.,
  having ordinate value $i+k$ (pay attention to the fact that the point $(x,y)$
  is at row $x$, i.e., ordinate value~$x$, and column~$y$, i.e., abscissa~$y$).
  In particular, we must have $p_n = (i',j')$, and the distance of each $p_k$
  to~$\pi_k$ is only due to the second coordinate, and it is at most~$1/2$.

  We now modify slightly our choice of the points $p_1, \ldots, p_n$ to obtain
  points $p_1', \ldots, p_n'$ and ensure that the latter points
  have an incident vertical edge pointing up, i.e., if $p_k =
  (i'',j'')$, we want to make sure that the vertex is adjacent to $(i''-1,
  j'')$. This is the case of every other vertex on each row, so it can be enforced by
  replacing each point $p_i$ not satisfying the condition by a left or right
  neighbor $p_i'$
  which does, i.e., incrementing or decrementing the first component. (It may be
  the case that only one of these two options is possible if we are close
  to the borders of the wall.) Specifically, we choose the neighbor which is closest to the
  segment.
  In so doing, the distance of the intermediate points to the segment becomes at
  most~$1$, achieved in the case where $\pi_k = p_k$ and $p_k$ is not suitable
  (see the case of $p_3$ on Figure~\ref{fig:wallapprox}).

  It suffices now to design a path in the wall graph that successively visits
  the points $p_0 = p_0', \ldots, p_n', p_n$ (the last step being one single
  edge, or no edge if $p_n = p_n'$).
  We do so inductively: from a vertex~$p_k'$, we
  first change the second component to be that of~$p_{k+1}'$ while keeping the
  first component unchanged, then we increase the second component (traversing
  the edge whose existence is guaranteed by the choice of $p_{k+1}'$ relative to
  $p_{k+1}$). The only
  exceptional case is the last: if the node $p_n$ reached is not a segment
  endpoint but a left or right neighbor, we finish by traversing the requisite
  edge. It is clear that this definition ensures that the path is simple. The
  path is illustrated again in Figure~\ref{fig:wallapprox}.

  What is more, the distance of this path to the segment is at most 2.
  Indeed, consider any choice of $k > 0$, and let us study how far away the path
  from~$p_{k-1}'$ to $p_k'$ strays from the segment. The segment connects $\pi_{k-1}$ to~$\pi_k$,
  and we know that the first coordinate of~$p_{k-1}'$ and $\pi_{k-1}'$ is $i+k-1$
  and the first coordinate of~$p_{k}'$ and $\pi_{k}'$ is $i+k$. So the segment is
  at distance at most~$1$ to any point in the wall graph whose first coordinate
  is $i+k-1$ or $i+k$ and whose second coordinate is between that of~$\pi_{k-1}$
  and that of~$\pi_k$. But we know the path from~$p_{k-1}'$ to~$p_k'$ visits only
  points satisfying these conditions except that the second coordinate
  of~$p_{k-1}'$ and of~$p_k'$ may each differ by at most~$1$ for that
  of~$\pi_{k-1}$ and of~$\pi_k$. So in total the distance of the visited
  vertices to the segment is at most $1+1 = 2$. This bound of~$2$ holds for
  the path from~$p_0$ to~$p_n'$, and the last edge added in case $p_n' \neq p_n$ is 
  not a problem because then we are clearly at distance $\leq 1$ from the
  segment. This concludes the proof.
\end{proof}

\section{Scaling up the drawing}
\label{sec:fixing}

\begin{figure}
  \centering
  \begin{tikzpicture}[scale=1.5]
    \tikzgrid{6}{4}
    \node[draw,circle] (a) at (0,0) {};
    \node[draw,circle] (b) at (4,2) {};
    \node[draw,circle] (c) at (4,4) {};
    \node[draw,circle] (d) at (6,2) {};
    \draw[orange] ($(a)+(0,0.15)$) -- ($(b)+(0,0.15)$);
    \draw[orange] ($(a)-(0,0.15)$) -- ($(b)-(0,0.15)$);
    \draw[very thick] (a) -- (b);
    \draw[orange] ($(a)+(0,0.2)$) -- ($(c)+(0,0.2)$);
    \draw[orange] ($(a)-(0,0.2)$) -- ($(c)-(0,0.2)$);
    \draw[very thick] (a) -- (c);
    \draw[very thick] (b) -- (c);
    \draw[very thick] (b) -- (d);
    \draw[orange] ($(b)+(0,0.15)$) -- ($(d)+(0,0.15)$);
    \draw[orange] ($(b)-(0,0.15)$) -- ($(d)-(0,0.15)$);
    \draw[orange] ($(b)+(.12,0)$) -- ($(c)+(.12,0)$);
    \draw[orange] ($(b)-(.12,0)$) -- ($(c)-(.12,0)$);
    \draw [blue,name path=rect] (-.85,-.85) rectangle (.85,.85);
    \draw [blue,name path=rect] (3.15,1.15) rectangle (4.85,2.85);
    \draw [blue,name path=rect] (3.15,3.15) rectangle (4.85,4.85);
    \draw [blue,name path=rect] (5.15,1.15) rectangle (6.85,2.85);
    \draw[red,dashed,ultra thick](3.12,2.87) circle (0.13);
    \draw[red,dashed,ultra thick](.85,.61) circle (0.13);
  \end{tikzpicture}
  \caption{Illustration of the global construction. We consider a box $B_v''$ centered on every
  vertex $v$ (in blue), and a distance margin $M_{u,v}$ around every segment
  $[u,v]$ corresponding to an edge~$\{u,v\}$ (in orange). We choose the box sizes
  to guarantee several things: (i) the boxes do not overlap;
  (ii) the margins around two different segments with the
  same endpoint do not intersect outside of the box of this endpoint (using
  Lemma~\ref{lem:boxdist}, see dashed red circle near the bottom left of
  figure);
  (iii) a box cannot intersect with
  the distance margin around a segment unless its vertex is an endpoint of that
  segment (using Lemma~\ref{lem:pass}, see dashed red circle in the middle of
  figure). This guarantees that the path approximations that stay close to each edge
  (Lemma~\ref{lem:line}) cannot intersect except within each of the boxes.}
  \label{fig:situation}
\end{figure}
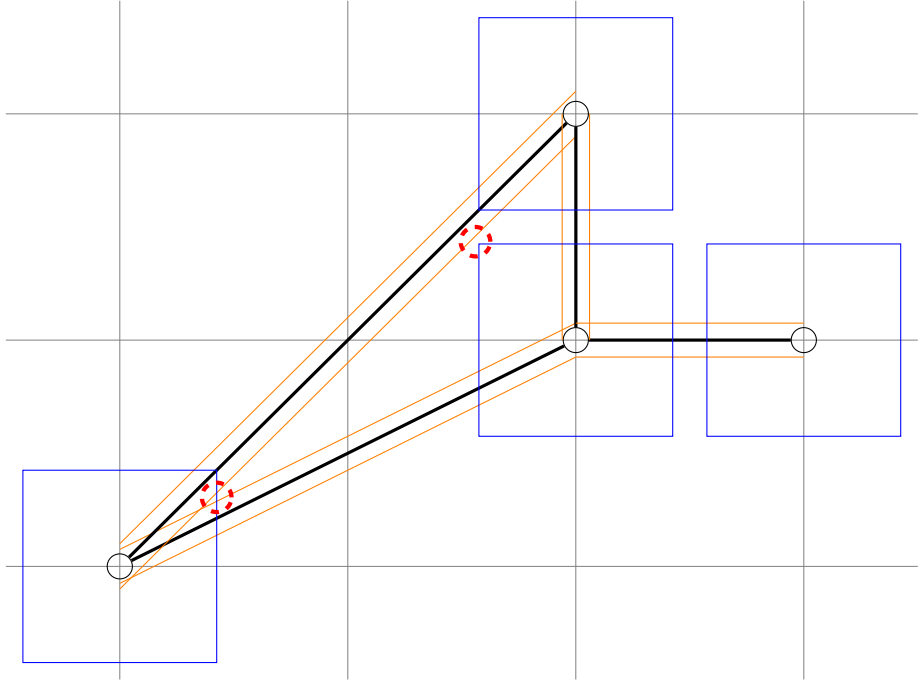

Let us recall what we showed in Section~\ref{sec:drawing}. 
Let us fix a box size
$\epsilon \colonequals \frac{1}{4n}$. We will consider boxes $B_v$ of radius~$\epsilon$
centered on each vertex (recall Definition~\ref{def:box}), and then scale up the
drawing.
For the scale-up, define the scale
factor $\sigma \colonequals 40 n^3$.
Take the drawing function~$\delta$ defined in Section~\ref{sec:drawing}, and multiply it
by~$\sigma$.
We can now consider the $(r,s)$-wall graph $W_{r,s}$ with $r = \sigma (n-2)$ and $s
= \sigma (n-2)$. The vertices of the drawing correspond to points of the wall graph
that are multiples of~$\sigma$.

For each vertex~$v$, we will consider the box $B_v'$ 
centered on~$v$ obtained from the box~$B_v$ after scaling up, whose radius is
$\sigma \epsilon = 10 n^2$.
We will also consider a larger box $B_v''$ centered on~$v$ with radius $\sigma \epsilon +
2$.
We will finally consider, for every segment $[u,v]$ corresponding to an edge
$\{u,v\}$, the \emph{margin} $M_{u,v}$ at distance 2 around $[u,v]$, consisting
of all points at Euclidean distance $\leq 2$ from the segment.
See Figure~\ref{fig:situation} for a summary of the situation.
We make four claims:

\begin{enumerate}[(i)]
  \item For any two distinct vertices $u \neq v$, the boxes $B_u''$ and $B_v''$ for~$u$ and~$v$ do not
    intersect. This is because
    $2 \sigma \epsilon + 4 < \sigma$.
  \item For any two segments $[u,v]$ and $[u,w]$ corresponding to edges $\{u,v\}$ and $\{u,w\}$
    sharing one vertex endpoint~$u$, the margins $M_{u,v}$ and $M_{u,w}$
    do not intersect except inside the larger box $B_u''$ 
    of~$u$. This is intuitively obtained by scaling up the bound from
    Lemma~\ref{lem:boxdist}. Specifically,
    consider any
    point $x$ in $M_{u,v} \cap M_{u,w}$. By definition, $x$ is at distance $\leq 2$ from both $[u,v]$
    and $[u,w]$, and the projections $v'$ and $w'$ of $x$ on $[u,v]$ and
    $[u,w]$ respectively are then at distance
    $\leq 4$ by the triangle inequality. Now, scaling up the bound from
    Lemma~\ref{lem:boxdist}, we know that any two points 
    respectively on segments $[u,v]$ and
    $[u,w]$ which are both outside the box $B_u'$ must be at distance at least
    $\frac{\sigma\epsilon}{2 n^2}$, i.e., at least~$5$. Thus, 
    at least one of the projections $v'$ and $w'$ must be inside the
    smaller box $B_u'$. Now, remember that 
    the point $x$ is at distance at most~$2$ from that projection, so
    we conclude that $x$ is inside the larger box $B_u''$, which concludes.
  \item For any segment $[u,v]$ of the drawing and point $w$ different from $u$ and $v$, the
    box~$B_w''$ does not contain any point of the margin $M_{u,v}$ of~$[u,v]$.
    This is obtained by scaling up the bound from Lemma~\ref{lem:pass}.
    Specifically, considering the point $w$ and the segment $[u,v]$, we know
    that the distance between $w$ and $[u,v]$ is at least
    $\frac{\sigma}{\sqrt{2}n}$, which is at least $20 n^2$. Now, the points of
    the box $B_w''$ are at distance $\leq \sqrt{2} \times (\epsilon \sigma+2)$ from~$w$,
    i.e., at distance $\leq \sqrt{2} \times (10 n^2+2)$. Further, the points
    of the margin are at distance $\leq 2$ from the segment. Now, we have
    $\sqrt{2} \times (10 n^2+2) + 2 < 20 n^2$ because $\sqrt{2} < 1.5$ and $n
    \geq 1$, 
    so we conclude that the box $B_w''$ and margin do not intersect.
  \item For two segments that do not share any endpoints, their margins do not
    intersect. Indeed, the segments do not intersect (this uses the fact that
    there are no edge crossings), and the
    distance between each endpoint pair is at least $\sigma$, hence, at
    least~$5$.
    (Note that, for any two segments that do not intersect, the minimum of the
    distance from one segment to another can always be achieved by taking one
    endpoint of one of the segments.)
\end{enumerate}

In summary, after scaling up, we obtain a situation similar to
Figure~\ref{fig:situation}. 
The vertices are at points whose coordinates are multiples of
the integer $\sigma$, and considering the box $B_w''$ of each vertex~$w$, and the
margin $M_{u,v}$ at distance 2 around each segment $[u,v]$, then (i) the boxes
are pairwise disjoint; (ii) the margins of two segments sharing an endpoint
intersect only within the box of that endpoint; (iii) the box on a vertex only
intersects with the margins of the segments where this vertex occurs; (iv) two
margins of segments that do not share an endpoint do not intersect.

Now, compute in polynomial time with Lemma~\ref{lem:line} a path for each edge,
that connects the endpoints of the edge (i.e., the vertices) while staying at
distance at most~$2$ from the segment that represents the edge, i.e., staying in
the margin of the segment corresponding to the edge.
The total time to compute these paths can be upper bounded by the total grid
size, which is $O(n^8)$, times the number of paths, which is $O(n^2)$, hence
an upper bound of $O(n^{10})$.

We know that the path for each segment remains in the margin at distance 2 of
that segment. We now claim that the paths of two segments can only share some intermediate
vertices if the segments share an endpoint $v$, and that then the shared intermediate
vertices are within the box~$B_v''$ of that endpoint. Indeed, if the segments do
not share any endpoints then their margins do not intersect by~(iv), and as the
paths are contained within the margins, they also share no vertices in that case. Now, if the
segments share an endpoint~$v$, the margins only intersect within~$B_v''$ by~(ii),
hence the same is true of the paths.

Thus, the only remaining part is to fix the embedding by
changing the paths within the boxes. We can do so
for each box independently, as the boxes are disjoint by~(i), and by~(iii) the box
$B_v''$ centered on a vertex~$v$ does not contain any
vertex for paths corresponding to edges that do not involve~$v$.

\section{Making the paths disjoint within the boxes}
For each vertex $v$ of the graph, let $G_v$ be the subgraph
induced by the 
vertices of the box $B_v''$ of radius $\sigma \epsilon + 2$ centered on~$v$.
Note that, because the coordinates of the box centers are multiples of the even integer
$\sigma$ and the box radius $\sigma\epsilon+2$ is an even integer, we know that 
$G_v$ is isomorphic to a $(2\sigma\epsilon+4,2\sigma\epsilon+4)$-wall graph.
For each edge $e = \{u,
v\}$, let $u_e$ and $v_e$ be respectively the last and first vertices of the path from~$u$ to~$v$ which
are in the boxes~$B_u''$ and~$B_v''$ respectively, and let~$P_e$ be the subpath
starting just after~$u_e$ and finishing just before~$v_e$.

We know that the paths $P_e$ and
$P_{e'}$ for two distinct edges $e$ and $e'$ do not intersect, because these
paths are outside the boxes of the vertex endpoints, they are within the margin
of the segment so cannot enter any other box by~(iii), and they cannot be within
the margin of any other segment by~(ii) and~(iv).

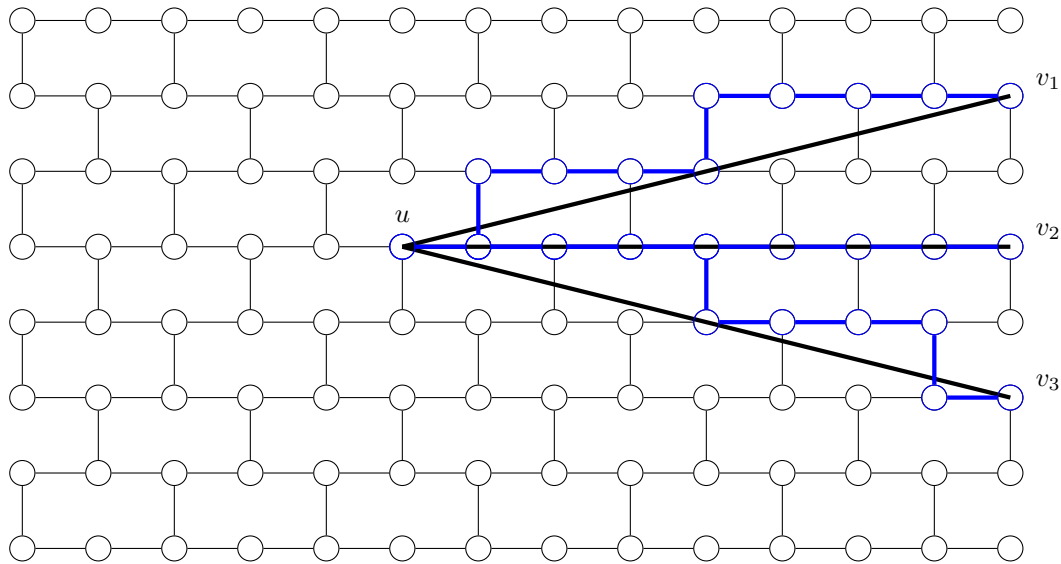
\begin{figure}
  \begin{tikzpicture}
    \wall{8}{14}{7}{13}{6}
    \draw[ultra thick] (6,5) -- (14, 3);
    \node (ql) at (14.5, 3.2) {$v_3$};
    \node (pl) at (14.5, 5.2) {$v_2$};
    \node (nl) at (14.5, 7.2) {$v_1$};
    \node (ul) at (6, 5.4) {$u$};
    \draw[ultra thick] (6,5) -- (14, 5);
    \draw[ultra thick] (6,5) -- (14, 7);
    \node[draw,circle,blue] (n0) at (6,5) {};
    \node[draw,circle,blue] (n1) at (7,5) {};
    \node[draw,circle,blue] (n2) at (8,5) {};
    \node[draw,circle,blue] (n3) at (9,5) {};
    \node[draw,circle,blue] (n4) at (10,5) {};
    \node[draw,circle,blue] (n5) at (11,5) {};
    \node[draw,circle,blue] (n6) at (12,5) {};
    \node[draw,circle,blue] (n7) at (13,5) {};
    \node[draw,circle,blue] (n8) at (14,5) {};
    \draw[ultra thick,blue] (n0) -- (n1) -- (n2) -- (n3) -- (n4) -- (n5) -- (n6) -- (n7) -- (n8);
    \node[draw,circle,blue] (p1) at (14,7) {};
    \node[draw,circle,blue] (p2) at (13,7) {};
    \node[draw,circle,blue] (p3) at (12,7) {};
    \node[draw,circle,blue] (p4) at (11,7) {};
    \node[draw,circle,blue] (p5) at (10,7) {};
    \node[draw,circle,blue] (p6) at (10,6) {};
    \node[draw,circle,blue] (p7) at (9,6) {};
    \node[draw,circle,blue] (p8) at (8,6) {};
    \node[draw,circle,blue] (p9) at (7,6) {};
    \node[draw,circle,blue] (p10) at (7,5) {};
    \draw[ultra thick,blue] (p1) -- (p2) -- (p3) -- (p4) -- (p5) -- (p6) -- (p7)
    -- (p8) -- (p9) -- (p10);
    \node[draw,circle,blue] (q0) at (6,5) {};
    \node[draw,circle,blue] (q1) at (7,5) {};
    \node[draw,circle,blue] (q2) at (8,5) {};
    \node[draw,circle,blue] (q3) at (9,5) {};
    \node[draw,circle,blue] (q4) at (10,5) {};
    \node[draw,circle,blue] (q5) at (10,4) {};
    \node[draw,circle,blue] (q6) at (11,4) {};
    \node[draw,circle,blue] (q7) at (12,4) {};
    \node[draw,circle,blue] (q8) at (13,4) {};
    \node[draw,circle,blue] (q9) at (13,3) {};
    \node[draw,circle,blue] (q10) at (14,3) {};
    \draw[ultra thick,blue] (q1) -- (q2) -- (q3) -- (q4) -- (q5) -- (q6) -- (q7)
    -- (q8) -- (q9) -- (q10);
  \end{tikzpicture}
  \caption{Choosing a vertex of the wall graph to embed the vertex $u$ of the
  graph with three neighbors $w_1, w_2, w_3$ such that the last vertices in the
  box in the paths approximating the edges $\{u, w_1\}$ and $\{u, w_2\}$ and
  $\{u, w_3\}$ are $v_1$ and $v_2$ and $v_3$. We show in black the beginning of
  the segments $[u,w_1]$ and $[u,w_2]$ and $[u,w_3]$, and show in blue the paths given by
  Lemma~\ref{lem:line}, which overlap. To fix the issue we simply connect $v_1$ and $v_2$ and
  $v_3$ as explained in the text and choose our target vertex for~$u$
  accordingly.}
  \label{fig:replacing}
\end{figure}

So let us now fix the rest of the paths, which are within the boxes.
Consider the paths from a vertex~$u$ to the vertices
of the form $u_e$ that are just before the paths~$P_e$.
There are at most~$3$ such vertices, because the graph is
degree-$3$. See Figure~\ref{fig:replacing} to understand why the paths from~$u$
to these three vertices may intersect. So we change completely these paths, in the process
possibly also changing the image of the vertex~$v$ in the embedding to be
another vertex of the box.

Let us state what we need:

\begin{lemma}
  \label{lem:paths}
  For any $(r,s)$-wall graph $W_{r,s}$ with $r \geq 1$ and $s \geq 2$, given three distinct
  vertices $p$, $q$, $r$, we can compute in time $O(r \times s)$ a vertex $x$
  of~$W_{r,s}$ and three vertex-disjoint paths between $x$ and $p$, and $x$
  and~$q$, and $x$ and~$r$.
\end{lemma}

\begin{proof}
  As $s \geq 2$, the graph is connected. Find in linear time an arbitrary simple
  path from~$p$ to~$q$. If $r$ is in the path, take $x=r$, which concludes.
  Otherwise, find in linear time an arbitrary simple path from~$r$ to~$p$, and
  stop as soon as it encounters a vertex of the previous path (which may be~$p$,
  $q$, or an intermediate vertex of that path). Take the
  encountered vertex to be~$x$, and the three vertex-disjoint paths connecting
  $x$ and $p$, $x$ and $q$, and $x$ and $r$, are then easy to find.
\end{proof}

Thus, for each vertex~$v$, we consider the box~$B_v''$ and the graph~$G_v$. If $v$ has strictly less
than 2 neighbors, there is nothing to do. If $v$ has two neighbors, letting
$v_1$ and $v_2$ be the last vertices before the paths $P_{e_1}$ and $P_{e_2}$
corresponding to the edge, we set the image of~$v$ to be~$v_1$ and pick an
arbitrary simple path from~$v_1$ to~$v_2$ in~$G_v$. If $v$ has three neighbors,
we use Lemma~\ref{lem:paths}
to find the image of~$v$ and the paths.
This process is in linear time in each
box, so in time
$O(n^8)$ overall where~$n$ is the number of graph vertices.

We have explained how to fix the embedding inside the boxes to ensure that there
are no intersections inside the boxes, and there are no intersections outside
the boxes, so this concludes the computation of the embedding and concludes the
proof of Theorem~\ref{thm:main}.

\bibliography{main.bib}
\end{document}